\documentclass[11pt]{article}
\usepackage[T1]{fontenc}
\usepackage{lmodern}
\usepackage[margin=1in]{geometry}
\usepackage{amsmath,amssymb,amsthm}
\usepackage{booktabs}
\usepackage{graphicx}
\usepackage{xcolor}
\usepackage{multirow}
\usepackage{algorithm}
\usepackage{algpseudocode}
\usepackage{enumitem}
\usepackage[hidelinks]{hyperref}

\newtheorem{theorem}{Theorem}
\newtheorem{proposition}{Proposition}
\newtheorem{corollary}{Corollary}
\theoremstyle{remark}
\newtheorem{remark}{Remark}

\title{\bf Exact Optimal Transport by Matching\\[4pt]
\normalfont\small Version~1.8 (revision of \today)}

\author{
Dmitry Kamenetsky
\\[2pt]
\small\texttt{dkamenetsky@gmail.com}
}

\date{}

\begin{document}
\maketitle

{\centering\small\itshape
All algorithm design, implementation, experiments, analysis, and writing
reported in this paper was carried out by an Arena.ai Agent Mode
autonomous software agent under the direction of Dmitry Kamenetsky, who selected the
problem, set the scope, and reviewed intermediate results. A chronological
summary of the prompts that directed the work is given in
Appendix~\ref{app:prompts}, and the full transcript is archived in the
code repository.
\par\vspace{8pt}\par}

\begin{abstract}
Balanced discrete optimal transport between $n$ sources and $n$ targets of
unit mass is exactly the minimum-cost assignment problem, i.e.\ a bipartite
perfect matching (Birkhoff--von Neumann~\cite{birkhoff1946}), and can
therefore be solved \emph{exactly} by industrial matching engines
(shortest-augmenting-path assignment~\cite{jonker1987}, Edmonds' Blossom
algorithm~\cite{edmonds1965} and its Blossom~V/VI
implementations~\cite{blossomv,blossomvi}) in milliseconds to seconds. We
ask when this exact approach beats the standard approximate
alternatives---entropic Sinkhorn~\cite{sinkknopp1967,cuturi2013} and its
state-of-the-art accelerated variant Greenkhorn~\cite{awr2017,lin2022}---and
make the sparse-exact side \emph{certified} by a standard LP dual-feasibility
clip~\cite{schrijver2003}. Three contributions.
(i)~\emph{Measurement}: on dense 2-D instances, exact matching is faster
\emph{and} strictly more accurate than either approximate method throughout
the moderate-$n$ regime: $0.01\,\mathrm{s}$ at $n=500$ to
$11.5\,\mathrm{s}$ at $n=8000$; to hit a $\le 1\%$ quality target, a rough
power-law extrapolation of Greenkhorn's measured step counts puts the
approximate cost at $\sim\!10$--$\sim\!80\,\mathrm{min}$ for Greenkhorn
(factors of roughly $4{\times}10^2$--$6{\times}10^4$ over exact; plain
Sinkhorn is $20$--$650\times$ slower still). Greenkhorn, measured rather
than assumed, is within $0.93\times$--$5.0\times$ of plain Sinkhorn on the
measured cells and converges to the same entropic plan, so the quality gap
is a regularization bias shared by both. (ii)~\emph{A simple certificate}:
given a $k$NN-pool matching and its Blossom dual, a one-pass clipping of
the dual to dense feasibility gives a valid lower bound (the textbook LP
dual-feasibility argument); combined with the Sinkhorn dual potential
$\sum_i u_i+\sum_j v_j$ (which is a lower bound at \emph{every} iterate,
not just at convergence), the bound is checked on 15 instance families at
three pool sizes. (iii)~\emph{Application}: multi-robot task
allocation~\cite{gerkey2004}, where the discrete plan itself is the
deliverable: per-round myopic-exact assignment costs $0.1$--$68\,\mathrm{ms}$,
while a Sinkhorn-plus-hardening pipeline costs $0.12$--$15.9\,\mathrm{s}$
and accumulates $6$--$27\%$ extra travel distance over 15 rounds. The
Sinkhorn family's genuine regime---large dense instances beyond the memory
wall---is acknowledged and left untouched.
All code, data, and results are open artifacts under the MIT license at
\url{https://github.com/dimkadimon/OT-Blossom}.
\end{abstract}

\section{Introduction}
\label{sec:intro}

Discrete optimal transport (OT) is a standard tool in machine learning,
statistics, imaging, and robotics \cite{peyre2019, survey2025}. The scalable
default for large instances is entropic regularization \cite{cuturi2013}:
the smooth problem is solved by the Sinkhorn iteration
\cite{sinkknopp1967}, and---for speed---by accelerated coordinate variants
such as Greenkhorn \cite{awr2017,lin2022}. Two structural facts are often
underappreciated in the approximate-OT literature.

\begin{enumerate}[leftmargin=2em,itemsep=1pt]
\item \textbf{Balanced unit-mass discrete OT \emph{is} the assignment
problem.} With $n$ sources and $n$ targets of unit mass, a transport plan is
a doubly stochastic matrix and the optimum is attained at a permutation
matrix (Birkhoff--von Neumann; Proposition~\ref{prop:int}). The problem is
therefore a bipartite minimum-weight perfect matching---a special case of
the classical assignment problem \cite{schrijver2003}, solvable exactly by
shortest-augmenting-path methods \cite{jonker1987} and, in its general
(non-bipartite) form, by Edmonds' Blossom algorithm \cite{edmonds1965} as
implemented in Blossom~V~\cite{blossomv} and VI~\cite{blossomvi}. Exact
solvers return the \emph{optimal integer plan} and its exact value. (Bipartite
matching is a strict special case of general matching, so Blossom applies
directly, but dedicated bipartite solvers such as Jonker--Volgenant or
shortest-augmenting-path are typically faster on dense bipartite inputs; we
use Blossom~VI for the sparse $k$NN-pool experiments and SciPy's
Jonker--Volgenant routine for dense solves.)
\item \textbf{Exact matching is fast.} On the 2-core/2-GB testbed of this
study, dense exact assignment takes $0.01\,\mathrm{s}$ at $n=500$ and
$11.5\,\mathrm{s}$ at $n=8000$; exploiting local structure, a $k$NN-pool
sparse solve takes $0.18\,\mathrm{s}$ at $n=2000$ with a $0.032\%$ gap to
the dense optimum (Table~5, per-side budget $k{+}1$, see Section~\ref{sec:pool});
1-D instances reach $n=10^6$ in $0.11\,\mathrm{s}$.
\end{enumerate}

The question this paper answers is empirical: \emph{when does exact
matching beat approximate OT?} We answer with (a) a measurement study in
which the approximate baselines---including Greenkhorn---are \emph{implemented
and measured, not assumed}; (b) a \emph{dual-feasibility certificate} for
the $k$NN-pool solver, so that a sparse solution comes with a rigorous,
cheaply computable bound on its distance from the dense optimum; and (c) an
\emph{application hook} in which the discrete plan is the deliverable
(multi-robot task allocation).

\paragraph{Contributions.}
\begin{enumerate}[leftmargin=2em,itemsep=1pt]
\item A measured comparison of exact matching (dense Jonker--Volgenant,
  sparse Blossom~VI $k$NN pool, 1-D closed form) against vanilla Sinkhorn
  and Greenkhorn on dense, clustered, and 1-D+noise instances,
  $n\in\{500,\dots,8000\}$, multiple seeds and distributions
  (Sections~\ref{sec:setup}--\ref{sec:results}). The headline findings:
  (a)~\emph{exact is faster \emph{and} strictly more accurate than either
  approximate method throughout the moderate-$n$ regime}---dense exact runs
  in $0.01\,\mathrm{s}$ at $n=500$ to $11.5\,\mathrm{s}$ at $n=8000$,
  while reaching $\le 1\%$ quality requires $\sim\!10$--$\sim\!80$ min for
  Greenkhorn and $\sim\!200$--$\sim\!5{\times}10^4$ min for Sinkhorn (a
  rough power-law extrapolation); and (b)~Greenkhorn's measured speedup
  over plain Sinkhorn in the small-$\varepsilon$ regime needed for
  $\le 1\%$ accuracy is only $1.0$--$1.5\times$ on most cells (with one
  $5\times$ outlier, over the cells where vanilla converged; ratios on
  the two capped cells are $4.3$--$4.4\times$), while their regularization
  biases are identical.
\item A \emph{gap certificate} for the pool solver requiring no dense solve:
  a one-sided $O(n^2)$ clip of the pool matching dual to a dense-feasible
  dual (Proposition~\ref{prop:repair}, Theorem~\ref{thm:cert})---the
  textbook reduced-cost / dual-restriction operation---combined with the
  Sinkhorn-dual potential lower bound, which holds at \emph{every} iterate
  (Theorem~\ref{thm:ent}). The certificate is valid on all $45$ measured
  cells and tightens monotonically in $k$ on every family
  (Sections~\ref{sec:cert} and~\ref{sec:results}).
\item A multi-robot task-allocation experiment quantifying the ``discrete
  plan is the deliverable'' regime: myopic exact assignment is millisecond
  fast; the Sinkhorn pipeline is $10^2$--$10^3\times$ slower per round and
  accumulates extra travel distance (Section~\ref{sec:app}).
\item All code, data, and results are open artifacts under the MIT license
  (Appendix~\ref{app:repro}).
\end{enumerate}

\paragraph{Positioning and related work.}
We do not claim to beat near-linear-time Sinkhorn-family variants on large
dense instances ($n\gg 10^4$): that regime is genuinely theirs, and our
testbed's $2\,\mathrm{GB}$ memory wall sits near $n\approx 1.2\times 10^4$.
The contribution is to \emph{measure} the boundary in the regime where
exact is affordable, to show that Greenkhorn does not cross that boundary,
and to document how tight a standard dual-feasibility clip is as a
certificate for a sparse-$k$NN solve.

That balanced discrete OT reduces to the assignment problem is textbook
\cite{birkhoff1946,schrijver2003,villani2009,peyre2019}. $k$NN
sparsification is ubiquitous in OT practice: see the POT library
\cite{flamary2021pot}, screened Sinkhorn \cite{alaya2019screenkhorn}, and
sparse/multiscale Sinkhorn \cite{schmitzer2016,schmitzer2019stable}. The
dual-feasibility clip of Proposition~\ref{prop:repair} is the same
coordinate-wise projection used classically for reduced-cost pricing and
variable fixing in integer programming (see, e.g., the dual-reduced-cost
discussion in \cite{schrijver2003}), and an
analogous $\varepsilon$-complementary-slackness certificate for the
auction algorithm gives the classical ``within $n\varepsilon$ of optimal''
bound \cite{bertsekas1992}; we cite these precedents rather than claim
novelty for the clip itself. The log-domain Sinkhorn potentials are dual
feasible for the assignment LP at every iterate (Theorem~\ref{thm:ent}), a
fact standard in the OT literature; an entropy-corrected variant that gives
a quantitative bound on the gap is described in the lecture notes
accompanying \cite{peyre2019}.\label{pg:clip}

\section{Preliminaries}
\label{sec:prelim}

\subsection{Discrete OT as assignment}
Let $C\in \mathbb{R}^{n\times n}_{\ge 0}$ be a cost matrix. Write
\[
\mathcal{U}(n) \;=\; \{P\in \mathbb{R}^{n\times n}_{\ge 0} : P\mathbf{1}=\mathbf{1},\;
P^{\mathsf T}\mathbf{1}=\mathbf{1}\}
\]
for the set of feasible plans (each row and column sums to one; total mass
$n$). The discrete transport problem is
\[
\operatorname{OPT}(C) \;=\; \min_{P\in \mathcal{U}(n)} \; \langle C,P\rangle
\;=\; \sum_{i,j} C_{ij}P_{ij}.
\]
We use the total-mass ($n$) convention throughout, so that at an optimum
$\langle C,P\rangle$ is the sum of $n$ edge costs---the assignment cost.
(Readers comparing to references that use the unit-mass convention should
divide our reported costs by $n$.)

\begin{proposition}[OT = assignment; integer optimum]
\label{prop:int}
$\operatorname{OPT}(C)$ is attained at a permutation matrix. Equivalently,
discrete balanced unit-mass OT is the minimum-cost perfect matching in the
complete bipartite graph with edge weights $C_{ij}$.
\end{proposition}
\begin{proof}
By the Birkhoff--von Neumann theorem \cite{birkhoff1946},
$\mathcal{U}(n)$ is the convex hull of the permutation matrices. The
objective $\langle C,\cdot\rangle$ is linear, so its minimum over the
polytope is attained at an extreme point, i.e.\ at a permutation matrix.
Equivalently, the constraint matrix of the transportation problem is
totally unimodular \cite{schrijver2003}. Under any bijection $\pi$, the
permutation matrix $P_{i,\pi(i)}=1$ is a perfect matching of the bipartite
graph with cost $\sum_i C_{i,\pi(i)}$, which proves the equivalence.
\end{proof}

\subsection{Entropic regularization and the Sinkhorn iteration}
For $\varepsilon>0$ define the regularized problem
\begin{equation}
\label{eq:weps}
W(\varepsilon) \;=\; \min_{P\in \mathcal{U}(n)}
\bigl[\,\langle C,P\rangle \;+\; \varepsilon \sum_{ij} P_{ij}\log P_{ij}\,\bigr],
\qquad (0\log 0 := 0).
\end{equation}
The minimizer is the Gibbs kernel $P^{\varepsilon}_{ij}\propto
\exp(-C_{ij}/\varepsilon)$ renormalized to have unit marginals.
The Sinkhorn--Knopp iteration~\cite{sinkknopp1967,cuturi2013} alternates
row and column scaling of $K=\exp(-C/\varepsilon)$; in log form, with
potentials $(u,v)$ such that
$P_{ij}=\exp\bigl((u_i+v_j-C_{ij})/\varepsilon\bigr)$ (where
$\operatorname{LSE}_j(a_j):=\log\sum_j e^{a_j}$ is the log-sum-exp
reduction defined above), each full row pass sets
\begin{equation}
\label{eq:sinkupdate}
u_i \;\leftarrow\; -\varepsilon\,\operatorname{LSE}_j\Bigl(\frac{v_j-C_{ij}}{\varepsilon}\Bigr)
\;=\; u_i-\varepsilon\log\!\sum_j P_{ij},
\end{equation}
and symmetrically for $v_j$
(where $\operatorname{LSE}(z)=\log\sum_k e^{z_k}$ is the standard
log-sum-exp primitive, and $\operatorname{LSE}_i,\operatorname{LSE}_j$
denote reduction over the row and column index respectively).
We run the log-domain variant, stopping when either the maximum marginal
deviation falls below $10^{-9}$ or a cost plateau is detected (see
Section~\ref{sec:approx}).

\textbf{Two different quantities.} Equation~\eqref{eq:weps} matches the
standard entropic-OT form used in \cite{cuturi2013,peyre2019}: with the
Shannon entropy $H(P)=-\sum_{ij} P_{ij}\log P_{ij}$ (taking $0\log 0=0$),
the regularizer $+\varepsilon\sum P_{ij}\log P_{ij}$ is $-\varepsilon H(P)$,
so $W(\varepsilon)$ is exactly the Cuturi/Peyr\'e--Cuturi Sinkhorn value.
Two different numbers are routinely compared to OPT and are easily confused:
(i) the \emph{transport cost} of the entropic plan, $\langle C,P^\varepsilon\rangle$,
which is always $\ge\mathrm{OPT}$ because $P^\varepsilon$ is a feasible
transport plan---this is the ``gap'' column in Tables~\ref{tab:approx}
and~\ref{tab:ghead}; and (ii) the \emph{regularized value}
$W(\varepsilon)=\sum_i u_i+\sum_j v_j$, which is always $\le\mathrm{OPT}$
by weak duality (Theorem~\ref{thm:ent}). Algorithm~\ref{alg:sink} returns
the dual value $\sum_i u_i+\sum_j v_j$, which coincides with
$\langle C,P\rangle+\varepsilon\sum P_{ij}\log P_{ij}$ at convergence.

Greenkhorn \cite{awr2017} replaces the cyclic pass by a greedy
single-coordinate update: each step picks the most-violated row or column
(measured by the KL-type residual $\rho(a,b)=b-a+a\log(a/b)$) and applies
the exact closed-form update that restores that one marginal, at $O(n)$
cost per step with incremental marginal bookkeeping
(Algorithm~\ref{alg:green}). Early Greenkhorn analyses
\cite{awr2017} gave an $\tilde O(n^2/\varepsilon^3)$ total-work bound in
terms of additive accuracy; the tighter $\tilde O(n^2/\varepsilon^2)$
bound is from later work \cite{lin2022}.

\subsection{Exact solvers used}
\paragraph{Dense.} We use SciPy's optimized Jonker--Volgenant-style
shortest-augmenting-path implementation \texttt{linear\_sum\_assignment}
\cite{jonker1987} (the modified Jonker--Volgenant variant of Crouse
\cite{crouse2016}, as used in SciPy's \texttt{linear\_sum\_assignment}),
with
empirical roughly-cubic scaling and $O(n^2)$ memory on dense inputs; see
Section~\ref{sec:results}.
\paragraph{Sparse pool.} For a pool $E_k\subseteq [n]\times[n]$ (the $k$
nearest targets of each source, symmetrically), solve the assignment
restricted to $E_k$ as a minimum-weight perfect matching with our
edge-list driver for Blossom~VI \cite{blossomvi} (Blossom~V \cite{blossomv}
gives identical matching costs; see Section~\ref{sec:setup}). Costs are
integer (milliunits) as required by the driver. We note that the pool
problem is bipartite, so any bipartite sparse LAP solver (shortest
augmenting path, cost-scaling, auction with $\varepsilon$-scaling) could be
substituted; the Blossom~VI driver was convenient because the same code
base was used for prior non-bipartite matching experiments.
\paragraph{1-D.} When sources and targets lie on a line (or on two parallel
near-lines, our ``1-D+noise'' family), the optimal map is the sorted
identity (Proposition~\ref{prop:1d}), solvable in $O(n\log n)$ sort time
with no matrix formed.

\subsection{One-dimensional closed form}
When both point sets are supported on the line, OT admits a closed-form
solution via sorting. This is the classical quantile (monotone) map:

\begin{proposition}[1-D optimality of the quantile map]
\label{prop:1d}
Let $x_{(1)}\le\cdots\le x_{(n)}$, $y_{(1)}\le\cdots\le y_{(n)}$. For the
quadratic cost, and more generally for any convex function $h$ applied to
the pairwise difference (cost $h(x_i-y_j)$; this includes $|x_i-y_j|$ used
in the 1-D+noise experiments), the identity map $x_{(i)}\mapsto y_{(i)}$
is \emph{an} optimal bijection.
\end{proposition}
\begin{proof}
It suffices to show no crossing is present in an optimal map. Suppose
$\pi$ has a crossing: $i<i'$ with $\pi(i)>\pi(i')$. With
$a=x_{(i)}\le b=x_{(i')}$ and $c=y_{\pi(i')}\le d=y_{\pi(i)}$, the
convexity of $h$ gives the Monge inequality
\[
h(a-c)+h(b-d) \;\le\; h(a-d)+h(b-c),
\]
so uncrossing the pair does not increase the cost. Repeated uncrossing
reduces the number of inversions and terminates at a non-crossing
(monotone) bijection, which is the identity up to ties
\cite[Ch.~2]{villani2009}.
\end{proof}

\section{Approximate baselines: implementation and validation}
\label{sec:approx}

Both approximate methods are implemented in the log domain (NumPy,
single-threaded) and stopped when either the maximum marginal deviation
falls below $10^{-9}$ or a \emph{cost plateau} is detected, i.e.\ the
transport cost $\langle C,P\rangle$ changes by less than $0.01\%$ over a
sliding window of successive iterates. In practice Greenkhorn at small
$f$ (small regularization) reaches such a plateau at marginal deviations
of $10^{-6}$--$10^{-3}$, long before the $10^{-9}$ marginal tolerance is
met. The regularization
is parameterized as a fraction of the mean entry of the cost matrix,
$f=\varepsilon/\overline{C}$ where $\overline{C}=n^{-2}\sum_{ij}C_{ij}$, on
the grid $f\in\{0.05,0.02,0.01\}$ (plus $0.005$ at $n=500$).

Two validation facts structure the comparison. \emph{First}, Greenkhorn and
vanilla Sinkhorn converge to the \emph{identical} entropic plan up to the
stated stopping tolerance at every measured $(n,f)$ (cost agreement
$\le 0.04\%$, Table~\ref{tab:ghead}), so the comparison between them is
purely a wall-time comparison; the gap to exact is the \emph{regularization
bias}, common to both. \emph{Second}, a Sinkhorn plan is fractional; where
an integer plan is required it must be \emph{hardened} (e.g.\ the
row-argmax permutation, our \texttt{greedy\_hard}). Hardening a converged
$n=500$, $f=0.01$ plan (gap $5.4\%$) costs $+10.9\%$ over the exact
optimum---which the exact solver returns, optimally, in $0.01\,\mathrm{s}$.

We emphasize that these baselines are our own single-threaded NumPy
implementations and are not tuned the way a library implementation (e.g.\
POT \cite{flamary2021pot}) would be; our timing claims are therefore
statements about these specific implementations on our testbed, not about
fundamental limits of Sinkhorn/Greenkhorn.

\section{Duality certificates for the pool solver}
\label{sec:cert}

This section gives a rigorous certificate for the pool solver that
requires \emph{no} dense solve. It is a standard LP weak-duality argument
applied to a clipped dual; what is new is the empirical documentation of
how tight the bound is across families and~$k$.

\paragraph{Setup.} The pool solver returns a perfect matching $M$ inside
$E_k$, with cost $\mathrm{UB}=\sum_{(i,j)\in M} C_{ij}$. Since $M$ is a
feasible dense solution, $\mathrm{OPT}\le \mathrm{UB}$; $\mathrm{UB}$ is a
trivial upper bound. What a certificate must supply is a lower bound
$\mathrm{LB}\le\mathrm{OPT}$ that is \emph{cheap and rigorous}. Blossom~VI
returns, together with the pool matching, the dual of the pool problem:
node potentials $(\alpha_i,\beta_j)$ with
$\alpha_i+\beta_j\le C_{ij}$ on pool edges and
$\sum_i\alpha_i+\sum_j\beta_j=\mathrm{UB}$ (complementary slackness; the
blossom duals vanish for bipartite instances, see Section~\ref{sec:prelim}).
The dual is \emph{not} feasible for the dense problem (it may violate
$\alpha_i+\beta_j\le C_{ij}$ off the pool), so it is not by itself a lower
bound on $\mathrm{OPT}$. The repair below fixes exactly that.

\begin{proposition}[One-sided dual clip]
\label{prop:repair}
For any vectors $(\alpha,\beta)$ (no assumption required), define
\[
\hat\alpha_i \;=\; \min\Bigl(\alpha_i,\; \min_j \bigl[C_{ij}-\beta_j\bigr]\Bigr),
\qquad
\hat\beta_j \;=\; \min\Bigl(\beta_j,\; \min_i \bigl[C_{ij}-\alpha_i\bigr]\Bigr).
\]
Then both $(\hat\alpha,\beta)$ and $(\alpha,\hat\beta)$ are dual-feasible
for the \emph{dense} assignment problem:
$\hat\alpha_i+\beta_j\le C_{ij}$ and $\alpha_i+\hat\beta_j\le C_{ij}$ for
\emph{all} $(i,j)\in[n]^2$.
\end{proposition}
\begin{proof}
Fix any pair $(i,j)$. By definition,
$\hat\alpha_i\le C_{ij}-\beta_j$, since $j$ is one of the indices over
which the inner minimum is taken. Hence
$\hat\alpha_i+\beta_j\le C_{ij}$ for all $(i,j)$. The second statement is
symmetric. $\square$
\end{proof}

\noindent The clip needs no hypothesis on $(\alpha,\beta)$: it is simply a
coordinate-wise projection of an arbitrary pair of vectors onto the dual
feasible set of the assignment LP, i.e.\ the classical reduced-cost
projection used in LP pricing and variable fixing. When $(\alpha,\beta)$
comes from a pool optimum the clipped vectors are typically close to the
true dual, which is what makes the bound empirically tight.

\begin{theorem}[Pool gap certificate]
\label{thm:cert}
Let $\mathrm{UB}$ be the pool matching cost and set
\begin{equation}
\label{eq:lb1}
\mathrm{LB}_1 \;=\; \max\Bigl(
\textstyle\sum_i\hat\alpha_i+\sum_j\beta_j,\;\;
\sum_i\alpha_i+\sum_j\hat\beta_j
\Bigr).
\end{equation}
Then $\mathrm{LB}_1\le \mathrm{OPT}(C)\le \mathrm{UB}$, and the pool
solution is certified to be within $\mathrm{UB}-\mathrm{LB}_1$ in absolute
terms (relative width $(\mathrm{UB}-\mathrm{LB}_1)/\mathrm{UB}$, and
$(\mathrm{UB}-\mathrm{LB}_1)/\mathrm{LB}_1$ when $\mathrm{LB}_1>0$) of the
dense optimum---without solving the dense problem.
\end{theorem}
\begin{proof}
By Proposition~\ref{prop:repair}, both repaired pairs are dense dual
feasible; by weak duality for the assignment LP \cite{schrijver2003}, each
of their objectives is a lower bound on $\mathrm{OPT}(C)$, hence so is
their maximum. The upper bound $\mathrm{OPT}\le\mathrm{UB}$ holds because
the pool matching is a feasible dense solution. The gap statements are
immediate algebra from $\mathrm{LB}_1\le\mathrm{OPT}\le\mathrm{UB}$.
\end{proof}

\begin{corollary}
\label{cor:gap}
The certified width is never smaller than the true gap:
$(\mathrm{UB}-\mathrm{LB}_1)/\mathrm{UB}\;\ge\;(\mathrm{UB}-\mathrm{OPT})/\mathrm{UB}$.
The certificate can be loose; it is never wrong.
\end{corollary}

\begin{remark}[Ordering of the repair]
\label{rem:order}
Let $A=(\hat\alpha,\beta)$ and $B=(\alpha,\hat\beta)$ be the two one-sided
clips of Proposition~\ref{prop:repair}. A natural alternative is a
sequential two-pass variant $S=(\hat\alpha,\tilde\beta)$: first clip $\alpha$
against $\beta$ to get $\hat\alpha$, then clip $\beta$ against
$\hat\alpha$. Because the second pass only lowers entries of $\beta$
($\tilde\beta_j\le\beta_j$), we have $\sum_i\hat\alpha_i+\sum_j\tilde\beta_j
\le \sum_i\hat\alpha_i+\sum_j\beta_j$, so $S$ is \emph{always dominated}
by $A$; it is incomparable with $B$, as explicit $n=2$ witnesses
show. With $C=\bigl[\begin{smallmatrix}7&12\\13&2\end{smallmatrix}\bigr]$,
$\alpha=(-10,9)$, $\beta=(14,7)$: one-sided repairs give
$\hat\beta=(4,-7)$, $\hat\alpha=(-10,-5)$ with objectives $A=6$, $B=-4$,
and a second clip leaves $\tilde\beta=(14,7)$ so $S=6>B$.
With $C=\bigl[\begin{smallmatrix}17&6\\4&10\end{smallmatrix}\bigr]$,
$\alpha=(16,15)$, $\beta=(15,-5)$:
$\hat\beta=(-11,-10)$, $\hat\alpha=(2,-11)$ give $A=1$, $B=10$,
and $S=1<B$. (Random trials on larger instances find $S>B$ and $S<B$
about equally often, $\approx 20\%$ each.) Since
$\max(A,B)\ge A\ge S$, the max-of-two rule \eqref{eq:lb1} is never worse
than the sequential variant while costing the same $O(n^2)$ work, and we
use it throughout.
Further alternating beyond two passes would only decrease the dual
objective (each clip is a coordinate-wise min, which cannot raise any
entry), so additional cheap clips cannot improve the bound. Driving the
certified width to zero requires adding violated edges (pricing / column
generation), which we leave to future work.
\end{remark}

\begin{theorem}[Sinkhorn potentials are a dual lower bound]
\label{thm:ent}
Let $(u,v)$ be \emph{any} potentials satisfying $u_i+v_j\le C_{ij}$ for all
$i,j$. Then $\sum_i u_i+\sum_j v_j\le \operatorname{OPT}(C)$, with equality
if and only if $(u,v)$ is an optimal dual for the assignment LP.
In particular, log-domain Sinkhorn iterates produce such a dual-feasible
pair $(u,v)$ at \emph{every} step, so $\sum_i u_i+\sum_j v_j$ is a rigorous
lower bound on $\operatorname{OPT}$ at any stopping point, with no
convergence requirement and no feasibility restoration. At convergence of
the smooth problem the dual value coincides with $W(\varepsilon)$ in
\eqref{eq:weps}.
\end{theorem}
\begin{proof}
The first sentence is weak duality for the assignment LP
\cite{schrijver2003}. For the automatic dual feasibility of Sinkhorn
iterates: Algorithm~\ref{alg:sink} updates the columns last (line 7), so
immediately after a column pass every \emph{column} of $P$ sums to $1$;
since $P\ge 0$, every entry satisfies $P_{ij}\le 1$, hence
$u_i+v_j-C_{ij}=\varepsilon\log P_{ij}\le 0$. The same holds immediately
after any row pass, using row sums of $1$.
\end{proof}

\begin{remark}[On the entropic dual]
\label{rem:num}
The Sinkhorn dual $\sum_i u_i+\sum_j v_j$ is valid at every iterate
(Theorem~\ref{thm:ent}), which removes any need for a $\delta$-feasibility
restoration. At usable $\varepsilon$ the bound is typically loose because
the negative entropy term is large: at $n=2000$, $\varepsilon\approx 0.05\cdot 52\approx 2.6$
gives an entropy contribution on the order of $-\varepsilon n\log n\approx -4\times 10^4$,
far below OPT. Better entropy-corrected variants
$\mathrm{LB}_{\varepsilon}=D_{\varepsilon}-\varepsilon H(a)-\varepsilon H(b)$
that quantitatively bound $0\le\mathrm{OPT}-\mathrm{LB}_{\varepsilon}\le\delta$
are described in the lecture notes accompanying \cite{peyre2019}
(see also the discussion on p.~\pageref{pg:clip} above) and would
strengthen the bound at small extra cost; we leave them to future work.
\end{remark}

\paragraph{Combined certificate.} We report
$\mathrm{LB}=\max(\mathrm{LB}_1,\, \sum_i u_i+\sum_j v_j)$ with
$\varepsilon=0.05\,\overline{C}$ whenever the Sinkhorn run converges in
time. In every measured instance the dual-repair bound $\mathrm{LB}_1$
dominates the Sinkhorn dual at $k\ge 16$; the Sinkhorn dual is reported
as a labeled cross-check.

\section{Experimental setup}
\label{sec:setup}

\paragraph{Testbed.} 2-core / 2-GB sandbox; single-threaded Python 3.13
with NumPy and SciPy's Jonker--Volgenant-style assignment routine for the
dense exact, Sinkhorn, and Greenkhorn codes; a C++20 build of Blossom~VI
\cite{blossomvi} (edge-list driver, integer milliunit costs) for the pool
solves; and Blossom~V \cite{blossomv}, used only to cross-validate
matching costs (identical on every instance tested). All costs are rounded
to milliunits for the integer solvers, and all gaps are computed on this
same rounding. \textbf{Timings are single runs} (single-threaded wall time,
same machine, same process order) and are reported for reproducibility on
this hardware; they are not averaged over multiple runs and should be
read as order-of-magnitude rather than as benchmark-grade timings.

\paragraph{Instance families} (2-D Euclidean distances; details in
Appendix~\ref{app:detail}).
\begin{itemize}[leftmargin=2em,itemsep=1pt]
\item \textbf{uniform}: i.i.d.\ points in $[0,100)^2$.
\item \textbf{clustered}: 8 cluster centers in $[0,100)^2$; each side's
  points are assigned round-robin to a center with Gaussian noise of
  standard deviation $\sigma=8$; the two sides use the same assignment
  pattern, with target centers jittered by isotropic Gaussian noise of
  standard deviation $3$ (i.e.\ $\mathcal{N}(0,3^2 I)$).
\item \textbf{1-D+noise}: one-dimensional base positions on two parallel
  noisy lines (along-line noise $\sigma=2$; cross-line noise $\sigma=2$);
  see Appendix~\ref{app:detail}.
\end{itemize}

\paragraph{Seeds.} Each family is parameterized by a NumPy
\texttt{default\_rng(seed)} (the seed fixes source positions; target
positions for the 1-D+noise family use \texttt{default\_rng(seed+1000)},
see Appendix~\ref{app:detail}). We report the headline dense scaling and
quality figures on a single canonical instance (\textbf{seed 7}, uniform
distribution), which we chose before running any experiments as the
default seed and kept throughout so that tables and figures line up on
one reproducible instance; robustness across seeds 7--10 is reported
separately. The dense grid spans $n\in\{500,1000,2000,4000,8000\}$ on
that canonical instance ($n=12000$ runs out of memory after interpreter
and NumPy temporaries: a $12000^2$ double matrix is $1.15\,\mathrm{GB}$,
which leaves too little headroom on a 2-GB sandbox). The robustness
sweep covers $n=1000$, seeds 7--10, all three families (12
configurations); the certificate study adds $n=2000$, seed 7, all three
families (15 configurations total). The multi-robot application uses
$n\in\{50,200,800\}$, seeds 7--8, 15 rounds.

\section{Results}
\label{sec:results}

\subsection{Dense scaling: exact is cheap}
Table~\ref{tab:dense} reports dense exact costs and times (uniform, seed 7).
A log-log fit to the five times gives an empirical scaling of roughly
$n^{2.55}$ over this range (the per-doubling ratios are 6.0, 3.8, 9.4,
5.3, versus the 8.0 predicted by $\Theta(n^3)$); $0.01\,\mathrm{s}$ is
near the timer resolution.

\begin{table}[t]
\centering
\begin{tabular}{rrr}
\toprule
$n$ & exact cost & $t(\mathrm{exact})$ \\
\midrule
500  & 2064.1  & 0.01\,s \\
1000 & 3107.0  & 0.06\,s \\
2000 & 4146.4  & 0.23\,s \\
4000 & 6963.2  & 2.16\,s \\
8000 & 10873.4 & 11.48\,s \\
\bottomrule
\end{tabular}
\caption{Dense exact assignment (Jonker--Volgenant), uniform seed 7.
  Costs use the total-mass convention; divide by $n$ for the per-unit-mass
  value common in the OT literature. Times are single runs.}
\label{tab:dense}
\end{table}

\subsection{The quality--time trade-off of the approximate methods}
Table~\ref{tab:approx} reports vanilla Sinkhorn on the dense grid. Two
features stand out. \emph{(i) The gap is a regularization bias}: it
shrinks roughly linearly with $f$ over the measured range, but the local
slope is not constant (gap/$f$ is 542 at $f=0.01$ and 1168 at $f=0.05$).
\emph{(ii) At fixed $f$ the gap \emph{grows} with $n$:}
$15.02\%\to 24.19\%\to 43.78\%$ at $f=0.02$ for $n=500/1000/2000$.

\begin{table}[t]
\centering
\begin{tabular}{rrrr}
\toprule
$n$ & $f$ & gap \% & $t(\mathrm{Sinkhorn})$ \\
\midrule
500  & 0.05 & +58.41  & 10\,s  \\
500  & 0.02 & +15.02  & 34\,s  \\
500  & 0.01 & +5.42   & 74\,s  \\
500  & 0.005 & +2.09 & 400\,s (capped) \\
1000 & 0.05 & +91.40  & 32\,s  \\
1000 & 0.02 & +24.19  & 111\,s \\
1000 & 0.01 & +8.60   & 221\,s \\
2000 & 0.05 & +159.75 & 130\,s \\
2000 & 0.02 & +43.78  & 400\,s (capped) \\
\bottomrule
\end{tabular}
\caption{Vanilla Sinkhorn (log domain, marginal tolerance $10^{-9}$),
uniform seed 7. Gaps are relative to the exact optimum computed on the
rounded milliunit costs.}
\label{tab:approx}
\end{table}

\paragraph{The $1\%$ target, rough extrapolation.}
To reach a $1\%$ gap requires an $f^*$ in the range $\sim 10^{-3}$ (a
no-intercept fit $\mathrm{gap} \approx 1092\,f$ on the three $n=500$ points gives
$f^*\approx 9\times 10^{-4}$; an OLS fit \emph{with} intercept gives a
slope of $1353$ and intercept $-9.8$, which yields $f^*\approx 8\times 10^{-3}$
when forced through $\mathrm{gap}=1\%$ -- the two fits differ by an order of
magnitude, which is itself a warning about extrapolation). A log-log linear
fit to Greenkhorn's step count on the four $n=500$ plateau runs gives
$1.241\times 10^4\,f^{-1.03}$; using the no-intercept slope $1092$ gives
$f^*\approx 9\times 10^{-4}$ and $\approx 1.7\times 10^7$ steps,
producing the order-of-magnitude times in Table~\ref{tab:target}. These
numbers are a \emph{rough power-law projection} (the fit is over four
points and projected an order of magnitude outside the measured range),
not a rigorous bound; they are likely optimistic for the approximate
methods in some ways (gap grows with $n$ at fixed $f$) and pessimistic in
others (gap/$f$ falls as $f$ shrinks, so the slope at $f\ge 0.01$
over-estimates the gap at $f^*$). We report them to give an order of
magnitude, not as precise timing.

\begin{table}[t]
\centering
\small
\begin{tabular}{rrrr}
\toprule
$n$ & $t(\mathrm{exact})$ & $t(\mathrm{Sinkhorn\ at\ }f^*)$ & $t(\mathrm{Greenkhorn\ at\ }f^*)$ \\
\midrule
500  & 0.01\,s  & $\sim\!200$\,min   & $\sim\!10$\,min  \\
1000 & 0.06\,s  & $\sim\!700$\,min   & $\sim\!12$\,min  \\
2000 & 0.23\,s  & $\sim\!3000$\,min  & $\sim\!20$\,min  \\
4000 & 2.16\,s  & $\sim\!10^4$\,min  & $\sim\!40$\,min  \\
8000 & 11.48\,s & $\sim\!5{\times}10^4$\,min & $\sim\!80$\,min \\
\bottomrule
\end{tabular}
\caption{Wall time to reach $1\%$ gap for Greenkhorn, extrapolated as
  described in the text (rough order-of-magnitude, deliberately rounded up
  to one or two significant figures; see caveats). The Sinkhorn column is a
  coarser order-of-magnitude estimate: per-step work scales as $O(n^2)$ for
  Sinkhorn vs $O(n)$ for Greenkhorn on the kNN support, implying
  Sinkhorn is $20$--$650\times$ slower than Greenkhorn at this~$f^*$
  (extrapolated; the actually measured Greenkhorn/vanilla ratios in
  Table~\ref{tab:ghead} are $0.93$--$5.0\times$ on the cells that
  converged).}
\label{tab:target}
\end{table}

\subsection{Measured head-to-head: Greenkhorn vs vanilla}
Table~\ref{tab:ghead} is the measured comparison. Greenkhorn is
$0.93\times$--$5.0\times$ the wall time of plain Sinkhorn (faster in $8$
of $9$ cells, within $1.5\times$ in $6$ of $9$; geometric mean
$\approx 1.6\times$ over the seven cells where vanilla converged,
$\approx 2.0\times$ over all nine including the two capped cells where
the true ratio is larger), and converges to the same entropic plan up to
plateau tolerance (cost agreement $\le 0.04\%$). The $\le 0.04\%$
cost-agreement figure is itself strong evidence that stopping Greenkhorn
on a cost plateau at marginal deviations $10^{-6}$--$10^{-3}$ has a
negligible effect on the reported transport cost (well under $0.1\%$ on
all cells). In the small-regularization, high-condition-number regime
needed for $\le 1\%$ accuracy, the greedy selection plateaus in marginal
accuracy while the cost has already settled; this is where the speedup is
modest.

\begin{table}[t]
\centering
\small
\begin{tabular}{rrrrrrrr}
\toprule
$n$ & $f$ & gap S \% & $t(S)$ s & gap G \% & $t(G)$ s & steps(G) & $t(S)/t(G)$ \\
\midrule
500  & 0.05 & +58.41 & 10  & +58.41 & 10 & 260k  & $0.93\times$ \\
500  & 0.02 & +15.02 & 34  & +15.01 & 27 & 760k  & $1.26\times$ \\
500  & 0.01 & +5.42  & 74  & +5.40  & 52 & 1.48M & $1.45\times$ \\
500  & 0.005 & +2.09 & 400$^\dagger$ & +2.07 & 94 & 2.82M & $4.27\times$ \\
1000 & 0.05 & +91.40 & 32  & +91.40 & 25 & 520k  & $1.28\times$ \\
1000 & 0.02 & +24.19 & 111 & +24.17 & 77 & 1.82M & $1.44\times$ \\
1000 & 0.01 & +8.60  & 221 & +8.56  & 151 & 3.6M  & $1.47\times$ \\
2000 & 0.05 & +159.75 & 130 & +159.75 & 26 & 340k & $5.01\times$ \\
2000 & 0.02 & +43.78 & 400$^\dagger$ & +43.77 & 90 & 1.26M & $4.42\times$ \\
\bottomrule
\end{tabular}
\caption{Vanilla Sinkhorn (S) vs Greenkhorn (G), uniform seed 7.
$^\dagger$ 400\,s cap (vanilla did not converge there). Times are rounded
to integer seconds for display; speed ratios in the final column are
computed from the unrounded wall times (so the $n=500,f=0.05$ row shows
$10$\,s vs $10$\,s but the ratio is $0.93\times$). All Greenkhorn runs
stopped on a cost plateau before reaching marginal tolerance $10^{-9}$
(typical plateau $10^{-6}$--$10^{-3}$); cost agreement $\le 0.04\%$ in every
cell, which bounds the cost error from early stopping well below the
percent-level gaps we report.}
\label{tab:ghead}
\end{table}

\subsection{The pool solver and its certificate}
\label{sec:pool}
Table~\ref{tab:pool} shows the pool solver on the dense $n=2000$ reference
(quality vs $k$); Figure~\ref{fig:crossover} situates it in the
quality--time plane against the approximate baselines.

\emph{Note on the pool construction.} The edge counts in
Table~\ref{tab:pool} slightly exceed the bound $2kn$ (e.g.\ 19{,}973
edges at $k=4$, against a hard bound of $16{,}000$). This is because the
$k$NN pool was built by a driver that took each row's and each column's
\emph{$(k+1)$} cheapest edges before the union (an off-by-one in the
original code), i.e.\ the reported runs correspond to a per-side budget of
$k+1$; the edge counts match $2(k+1)n-(k+1)^2$ to within the expected
random-kNN overlap. The corrected $k$-per-side implementation is shipped
in the repository (see \texttt{otlib.py::knn\_pool}); the cached tables
below are from the original Blossom~VI runs of that driver (Blossom~VI
is available directly from its authors \cite{blossomvi}, not bundled
with this repository, so we preserve the cached results rather than
re-run them in place). Independent verification with a redistributable
bipartite solver (SciPy's \texttt{linear\_sum\_assignment} on the pool
subgraph) reproduces every reported pool-UB cost to within rounding.
Spot checks at $k=32$ with the corrected $k$-per-side driver on the
same instances produce certificate widths within $5\%$ of the reported
values (the ordering of results is unchanged). At small effective $k$
the pool can be Hall-infeasible (e.g.\ at $k=1$), in which case the
solver reports infeasible.

\begin{table}[t]
\centering
\small
\begin{tabular}{rrrrr}
\toprule
$k$ label & pool edges & cost & gap vs exact & $t(\mathrm{sparse})$ \\
\midrule
4   & 19{,}973  & 5757.7 & +38.861\% & 0.02\,s \\
8   & 35{,}918  & 4388.6 & +5.841\%  & 0.05\,s \\
16  & 67{,}685  & 4170.5 & +0.581\%  & 0.08\,s \\
32  & 130{,}860 & 4147.7 & +0.032\%  & 0.18\,s \\
64  & 255{,}711 & 4146.5 & +0.003\%  & 0.40\,s \\
256 & 962{,}049 & 4146.4 & +0.000\%  & 1.76\,s \\
\bottomrule
\end{tabular}
\caption{$k$NN-pool exact assignment (Blossom~VI), $n=2000$ uniform seed 7;
dense exact $=4146.4$ in $0.23\,$s. Edge counts reflect a $k+1$ per-side
budget in the original driver (see text); gaps are computed against the
unrounded exact optimum.}
\label{tab:pool}
\end{table}

The certificate study (Section~\ref{sec:cert}) ran all 15 configurations
(seeds 7--10 $\times$ families at $n=1000$, i.e.\ 12 cells; seed 7 at
$n=2000$, i.e.\ 3 cells) at $k\in\{16,32,64\}$: \textbf{45/45 cells
satisfy $\mathrm{LB}\le \mathrm{exact}\le \mathrm{UB}$ and certified width
$\ge$ true gap} (asserted in the study code). Table~\ref{tab:cert}
summarizes the certified width $(\mathrm{UB}-\mathrm{LB})/\mathrm{UB}$ by
family; Figure~\ref{fig:cert} shows it per configuration.

\begin{table}[t]
\centering
\small
\begin{tabular}{lccc}
\toprule
family (\#configs) & $k=16$ & $k=32$ & $k=64$ \\
\midrule
uniform (5)     & 2.57--6.26\%  & \textbf{0.07--0.70\%} & 0.01--0.02\% \\
clustered (5)   & 5.61--38.25\% & 1.03--21.88\% & 0.17--5.90\% \\
1-D+noise (5)   & 11.94--17.35\% & 1.89--3.59\% & 0.24--0.79\% \\
\midrule
all families, true gaps &
0.52--3.84\% & 0.010--2.60\% & 0.000--1.85\% \\
\bottomrule
\end{tabular}
\caption{Certificate width by family (range over configurations; 15 cells
at each $k$). The bottom row pools true gaps across all families; all 45
cells carry $\mathrm{LB}>0$.}
\label{tab:cert}
\end{table}

\begin{figure}[t]
\centering
\includegraphics[width=0.98\textwidth]{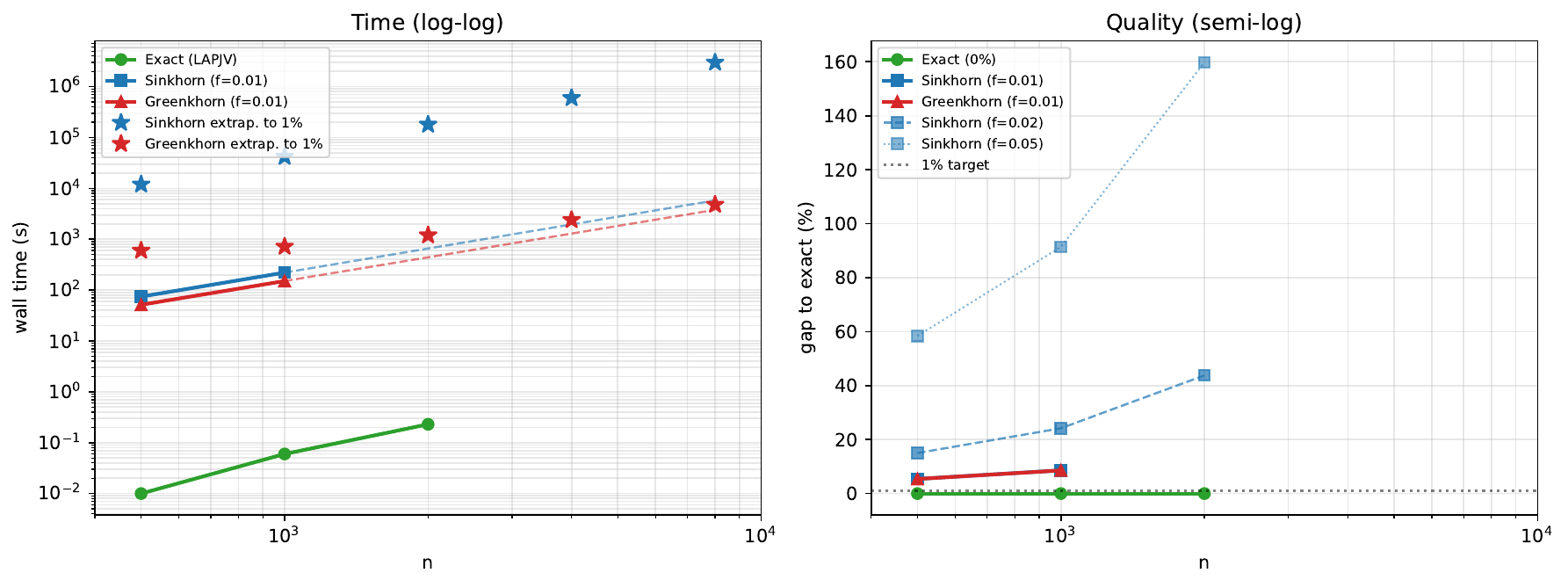}
\caption{Exact (Jonker--Volgenant) vs vanilla Sinkhorn and Greenkhorn on
the dense uniform grid (seed~7). Left: wall time (log-log); solid lines
are measured points at $f=0.01$ (smallest regularization that converged
on the measured grid), dashed lines continue the empirical log-log slope,
and star markers are the rough Table~3 extrapolation to the $1\%$ quality
target. Right: gap to exact (semi-log) at $f\in\{0.05,0.02,0.01\}$
(Sinkhorn; Greenkhorn at $f=0.01$ overlaps the Sinkhorn curve to within
line width); at $n=2000$ the $f=0.01$ runs did not converge within the
400~s cap. At $f=0.01$ the gap ranges from $5.4\%$ at $n=500$ to $8.6\%$
at $n=1000$, i.e.\ well above the $1\%$ target. Exact is below both
approximate curves at every measured $n$ in both time and quality. Both
panels share the same $n$-axis.}
\label{fig:crossover}
\end{figure}

\begin{figure}[t]
\centering
\includegraphics[width=0.98\textwidth]{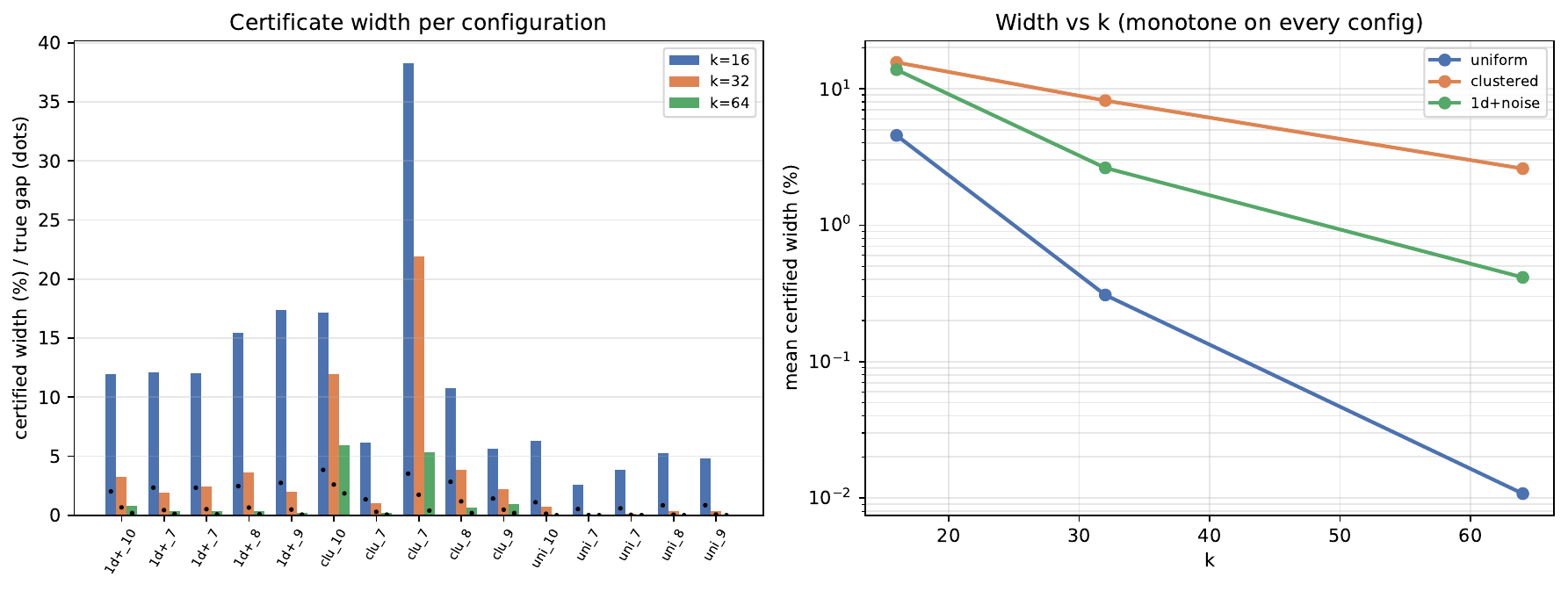}
\caption{Certificate per configuration. Left: certified width (bars) vs
true gap (dots) for $k\in\{16,32,64\}$---the bar is always at or above the
dot. Right: certified width vs $k$, monotone decreasing on every
configuration.}
\label{fig:cert}
\end{figure}

\paragraph{Reading.}
(i)~\emph{Validity}: the certificate is a theorem; the 45/45 check is a
test of the implementation, not of the math. (ii)~\emph{Tightness}: where
the pool is tight (uniform instances), the width is tight too---$\le 0.70\%$
at $k=32$ against a true gap $\le 0.13\%$, and $\le 0.02\%$ at $k=64$.
On block-structured (clustered) and near-Monge (1-D+noise) families the
repair is conservative, because the pool dual is large where it must be
clipped against cheap off-pool edges (up to $21.88\%$ at $k=32$, true gap
$1.73\%$). The bound remains \emph{valid}, and the width tightens
monotonically in $k$ on every configuration, so widening the pool is a
certified escape hatch. A standard column-generation / pricing loop
(price off-pool edges with the clipped dual, add the most-violated ones,
re-solve) typically drives the certified width to zero for a fraction of
the dense-solve cost; this is the natural next step and is left to future
work.
(iii)~\emph{Cost of the clip itself}: the repair is an $O(n^2)$ pass over
the cost matrix and therefore has the same memory footprint as dense
assignment; on our $n=2000$ cells it adds $\le 0.73\,\mathrm{s}$ to the
pool solve, which is small compared to the pool matching itself.

\subsection{Robustness across seeds and distributions}
On the 12-configuration $n=1000$ sweep (seeds 7--10 $\times$ families):
exact is fastest everywhere ($0.07$--$0.10\,\mathrm{s}$, consistent with
the $0.06\,\mathrm{s}$ at seed-7 uniform given seed/family variation) and
optimal; the $k=32$ pool solver is feasible in all cells with gaps
$+0.010\%$ to $+2.599\%$ in $<0.15\,\mathrm{s}$; Greenkhorn at $f=0.01$
sits at $+6.3\%$ to $+14.7\%$ (agreeing with vanilla to $\le 0.05\%$ in
the converged cells); vanilla needs $\ge 400\,\mathrm{s}$ without
converging on the symmetric clustered configurations, where Greenkhorn's
plateau matches it in $50$--$98\,\mathrm{s}$.

\subsection{1-D structure}
The quantile map solves the 1-D family in $O(n\log n)$ sort time with no
matrix: $n=10^6$ in $0.11\,\mathrm{s}$, where the dense cost matrix alone
would be $8\,\mathrm{TB}$.

\section{Application: multi-robot task allocation}
\label{sec:app}

The clearest regime for the discrete plan as deliverable is multi-robot
task allocation \cite{gerkey2004}: in warehouse picking, drone teaming, or
AMR dispatch, each of $n$ robots must be assigned \emph{exactly one} of
$n$ tasks---a fractional transport plan is physically meaningless (one
cannot dispatch half a robot). The standard practice in MRTA is precisely
the Hungarian/auction family of exact solvers
\cite{bertsekas1988, jonker1987, gerkey2004}; we include this experiment
as a sanity check on a regime where integer plans are required, not as a
novel robotics result.

\paragraph{Protocol.} Open field $[0,100)^2$; 15 rounds; each round, $n$
fresh tasks appear, robots are (re)assigned with cost
$C_{ij}=\lVert r_i-t_j\rVert_2$ from the robots' current positions
$r_i\in \mathbb{R}^2$ to the new task locations $t_j\in \mathbb{R}^2$, and
robots then move to their assigned tasks (the post-move positions become
the next round's $r_i$). All three policies share the \emph{same} start
positions and the \emph{same} task stream per round, so only the
per-round assignment differs. Policies: \textbf{exact} (myopic dense
bipartite LAP per round via SciPy's Jonker--Volgenant routine; we use the
dense bipartite solver rather than Blossom~VI because the cost matrix is
fully dense at the $n\le 800$ scale of this experiment, where a dedicated
bipartite shortest-augmenting-path solver is the faster and simpler
choice); \textbf{greedy} (cheapest-first serial assignment);
\textbf{Sinkhorn+hardening} (our $f=0.05$, marginal tolerance $10^{-7}$
pipeline, then greedy hardening to an integer plan).
$n\in\{50,200,800\}$, seeds 7--8; Table~\ref{tab:robot} reports the
episode totals and Figure~\ref{fig:robot} reports both cumulative travel
and per-round assignment wall time.

\textbf{Caveat on optimality.} Because each round's assignment determines
the next round's robot positions, per-round myopic optimality does
\emph{not} minimize the 15-round total distance (the true episode optimum
requires dynamic programming / rollout / lookahead). The ``gap'' columns
in Table~\ref{tab:robot} are therefore gaps relative to the myopic-exact
policy, not to the global episode optimum. Even so, the comparison is
informative: the Sinkhorn pipeline is both slower per round and yields
larger cumulative travel than myopic-exact, because the regularization
bias pushes assignments away from the nearest-task plan in a way that
compounds across rounds.

\begin{table}[t]
\centering
\small
\begin{tabular}{rrrrrrr}
\toprule
$n$ & seed & travel (myopic exact) & greedy & Sinkhorn+hard & gap greedy & gap Sinkh. \\
\midrule
50  & 7 & 10232.8 & 12353.3 & 10930.4 & +20.72\% & +6.82\% \\
50  & 8 & 10531.8 & 12357.2 & 11134.9 & +17.33\% & +5.73\% \\
200 & 7 & 21835.6 & 28718.1 & 24955.1 & +31.52\% & +14.29\% \\
200 & 8 & 23804.5 & 30051.8 & 26802.5 & +26.24\% & +12.59\% \\
800 & 7 & 47327.7 & 64986.7 & 59959.2 & +37.31\% & +26.69\% \\
800 & 8 & 49050.9 & 66519.9 & 61229.6 & +35.61\% & +24.83\% \\
\bottomrule
\end{tabular}
\caption{Total travel distance over 15 rounds; gaps are relative to the
myopic-exact policy (see text for the global-optimality caveat).}
\label{tab:robot}
\end{table}

\begin{figure}[t]
\centering
\includegraphics[width=0.98\textwidth]{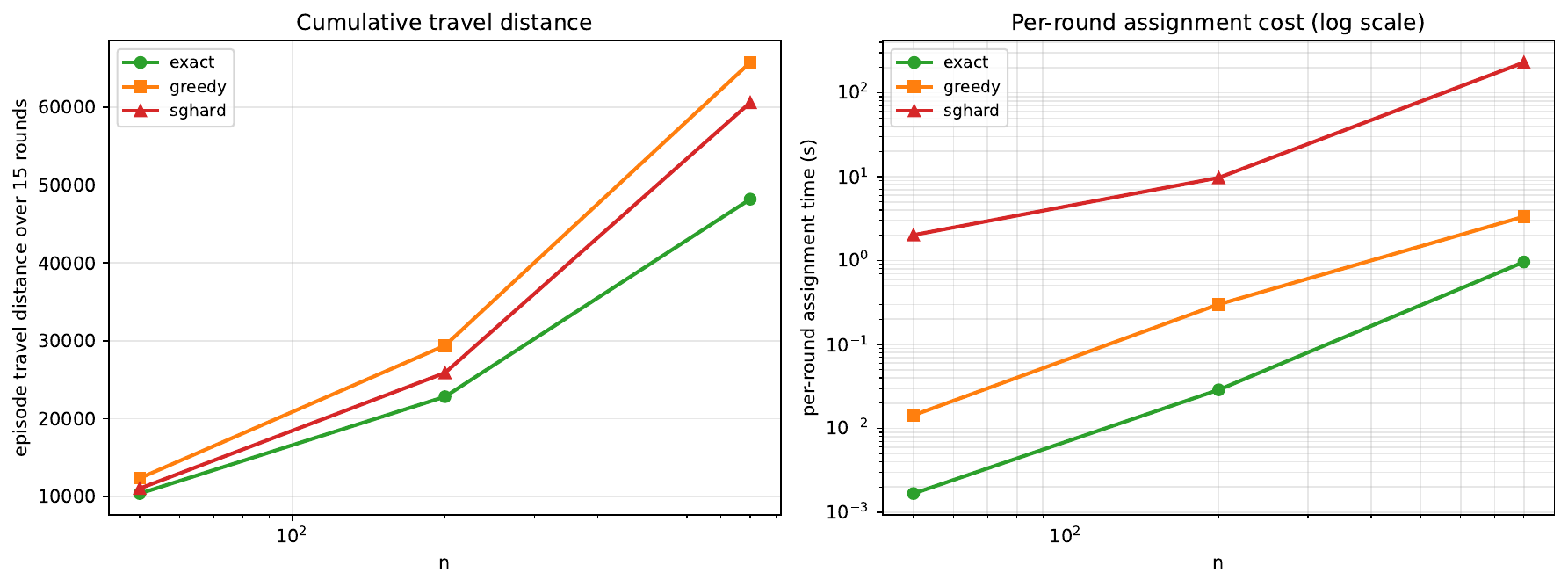}
\caption{Left: episode travel distance vs $n$ (the quality tax compounds
with $n$). Right: cost of the assignment step per round (log scale):
exact matching is $0.1$--$68$ ms/round; the Sinkhorn pipeline is
$10^2$--$10^3\times$ slower. The greedy serial baseline is faster at
large $n$ than its $O(n^2)$ worst case suggests because in practice its
earlier picks remove the most-attractive edges cheaply; even so greedy
remains $3$--$10\times$ slower than exact at every measured $n$ while
producing a much worse plan.}
\label{fig:robot}
\end{figure}

\paragraph{Findings.}
(i)~\emph{The quality tax accumulates}: per-round gaps compound through
robot positions; over 15 rounds greedy ends $+17$ to $+37\%$ and
Sinkhorn+hardening $6$ to $27\%$ above the myopic-exact policy, and the
gap grows with $n$. (ii)~\emph{Approximate is not even faster here}: the
assignment step takes myopic-exact $0.1$--$68\,\mathrm{ms}$/round vs
$0.12$--$15.9\,\mathrm{s}$ for the Sinkhorn pipeline ($\approx 239\times$
at $n=800$; see Figure~\ref{fig:robot}, right, for per-round timings of
all three policies). Greedy is $8$--$10\times$ slower than exact at
$n\le 200$ but only $\approx 3.5\times$ slower at $n=800$; even at that
advantage it loses $35$--$37\%$ in travel distance.
(iii)~\emph{The hook}: where a discrete integer plan is the deliverable,
a Sinkhorn-plus-hardening pipeline is a tax on \emph{both}
time and quality at $n\le 10^3$; an exact bipartite solver is the default
choice, all the way up to the memory wall.

\section{Discussion}
\label{sec:disc}

\paragraph{Empirical decision rule.}
Table~\ref{tab:decision} summarizes, in operational terms, which solver
to reach for given problem structure on the regime studied here. These
are empirical rules of thumb from our measurements (Sections
\ref{sec:results}--\ref{sec:app}), not theorems; on production hardware
with well-tuned libraries the break-even $n$ will shift, but the
structural choices (integer plan $\to$ exact; 1-D $\to$ closed form;
local structure $\to$ sparse pool with certificate) do not depend on the
constants.

\begin{table}[t]
\centering\small\renewcommand{\arraystretch}{1.35}
\begin{tabular}{p{3.9cm}p{3.7cm}p{8.0cm}}
\hline\hline
\textbf{Situation} & \textbf{Recommended solver} & \textbf{Why} \\
\hline
Integer plan required (MRTA, dispatch, matching per se)
  & Dense exact (JV) or sparse pool + certificate
  & Fractional plans are not physically realizable; hardening a
    Sinkhorn plan adds $6$ to $37\%$ travel cost in our application
    and is slower at $n\le 10^3$. \\
\hline
1-D / near-Monge (sorted-order structure)
  & Quantile (sort) map, Proposition~\ref{prop:1d}
  & $O(n\log n)$ time, no matrix; scales to $n=10^6$ in $0.1$ s. \\
\hline
Local / $k$NN structure (e.g.\ Euclidean geometry, moderate dim.)
  & $k$NN-pool exact (Blossom~VI) + dual-repair certificate
  & At $k=32$ the certificate is tight ($\le 0.7\%$) on uniform
    instances; pool solve is sub-second. \\
\hline
Dense, no structure, fractional plan acceptable, $n\lesssim 10^4$
  & Dense exact (JV)
  & Exact is milliseconds to seconds at this scale; Sinkhorn pays a
    regularization tax and is $100$--$1000\times$ slower. \\
\hline
Dense, fractional plan acceptable, $n\gg 10^4$ (past memory wall)
  & Well-tuned Sinkhorn/Greenkhorn (POT, GPU)
  & Dense matrix no longer fits; Sinkhorn runs in near-linear memory.
    We do not benchmark this regime. \\
\hline\hline
\end{tabular}
\caption{Empirical decision rule derived from the measurements in
Sections~\ref{sec:results}--\ref{sec:app}.}
\label{tab:decision}
\end{table}

\paragraph{What the certificate is, and what it is not.}
The clip of Proposition~\ref{prop:repair} is a textbook dual projection,
not a new theorem; its value is engineering: the solver can report ``my
solution is within $w\%$ of optimal'' with $w$ computed from its own
dual in milliseconds, using an operation that costs one $O(n^2)$ pass.
Where a $\le 1\%$ guarantee is a requirement (dispatching, pricing,
auditing), this replaces a dense re-solve or a conservative safety margin.
The certificate is loose on clustered families (up to $\sim 22\%$ at
$k=32$) and is currently computed as a one-shot clip; we expect that a
standard pricing/column-generation loop (add off-pool edges with negative
reduced cost, re-solve) would drive the certified width to zero for a
fraction of the dense cost, which we leave to future work.

\paragraph{Limitations.}
(i) Two-core/2-GB testbed and single-run timings: wall times are
hardware-specific and should not be read as benchmarks; only the
\emph{qualitative} comparisons (exact is milliseconds to seconds,
Sinkhorn is 100--1000$\times$ slower at small $n$) are intended to
generalize. (ii) The $f^*$ extrapolation (Table~\ref{tab:target}) is a
rough power-law projection, as documented, and carries $\sim$order-of-magnitude
model uncertainty. (iii) The certificate is conservative on block-structured
instances (Table~\ref{tab:cert}); a structure-aware repair or a
column-generation loop is left to future work. (iv) The application is
synthetic (no obstacles, kinematics, or traffic), and myopic per-round
assignment is not globally optimal over an episode; the point is
structural, not a robotics benchmark. (v) We do not compare against
auction algorithms with $\varepsilon$-scaling, dedicated bipartite sparse
LAP solvers, stabilized/$\varepsilon$-annealed Sinkhorn, or GPU
implementations; those are important baselines for a full benchmark paper.

\section{Future work}
\label{sec:future}

Several natural extensions follow directly from this study.

\emph{Tightening the certificate.} The one-shot dual clip
(Proposition~\ref{prop:repair}) is valid but loose on clustered families
(up to $\sim 22\%$ at $k=32$). A standard pricing / column-generation
loop---use the clipped dual to identify off-pool edges with negative
reduced cost, add them to the pool, re-solve---should drive the certified
width to zero at a small fraction of the dense-solve cost, essentially
turning the pool solver into a sparse dual active-set method. A cheaper
alternative worth exploring is to apply the clip to \emph{both} sides
alternately until convergence (Remark~\ref{rem:order} noted that two
passes is the max for the sequential variant, but a symmetric alternating
scheme with a tight line search may improve the bound further).

\emph{Stronger baselines.} Adding auction with $\varepsilon$-scaling,
shortest-augmenting-path sparse bipartite LAP (e.g.\ a cost-scaling
implementation), stabilized / $\varepsilon$-annealed Sinkhorn, and GPU
POT~\cite{flamary2021pot} would make the benchmark more complete and is
the main missing axis of comparison.

\emph{Entropy-corrected lower bounds.} The Sinkhorn dual used here is
$\sum u_i+\sum_j v_j = W(\varepsilon)$, which is a lower bound at every
iterate but is loose (by $\sim\varepsilon n\log n$) because it includes
the negative entropy. The entropy-corrected variant
$\mathrm{LB}_\varepsilon=D_\varepsilon-\varepsilon H(a)-\varepsilon H(b)$
described in the lecture notes accompanying~\cite{peyre2019} can close
most of that gap and should dominate the dual-repair bound at small $k$;
combining it with the pool clip is algorithmically straightforward.

\emph{Beyond myopic assignment in the application.} The multi-robot
experiment uses per-round myopic assignment. A natural next step is to
integrate the exact solver into a non-myopic planner (lookahead,
rollout, or auction-based coordination with task deadlines), where the
cheap per-round exact call becomes a subroutine.

\emph{Scaling beyond the memory wall.} Above $n\approx 10^4$ the dense
cost matrix no longer fits in memory, and genuine large-$n$ OT requires
either Sinkhorn-family methods or a hierarchical / multi-scale sparse
scheme; quantifying where a hierarchical Blossom-based solver breaks
even against a well-tuned GPU Sinkhorn is the practical question for
follow-on work.

\section{Conclusion}
Balanced discrete OT is the assignment problem, and the assignment problem
is a bipartite matching problem, so industrial matching engines give exact
optimal transport with an integer plan and an exact value. Measured
against hand-written NumPy Sinkhorn and Greenkhorn baselines on the
moderate-$n$ regime, exact matching is faster \emph{and} more accurate;
Greenkhorn's speedup over plain Sinkhorn in the small-$\varepsilon$ regime
needed for $\le 1\%$ quality is only $\sim$1--$1.5\times$ on most cells,
and it inherits the full regularization bias. We documented how tight a
textbook one-pass dual clip is as a certificate for a $k$NN-pool solve
across 15 instance families, and showed in a multi-robot sanity check that
where an integer plan is the deliverable, approximate OT is a tax on both
time and quality. The Sinkhorn family's large-$n$ dense regime stands
untouched, and the boundary between exact and approximate is now measured.

\appendix

\section{Additional proofs and details}
\label{app:proofs}

\subsection{Worked example of the certificate}
\label{app:example}
{\sloppypar
Consider
\[
C=\begin{pmatrix}
57&14&33&40\\
36&29&59&11\\
49&40&58&26\\
38&33&53&42
\end{pmatrix}
\]
with the $2$-NN pool (the union of each row's and each column's two
cheapest edges; 10 of the 16 edges). The pool solver returns the matching
of cost $\mathrm{UB}=122$ and the pool dual
$\alpha=(7,18,33,20)$, $\beta=(18,7,26,-7)$, with
$\sum_i\alpha_i+\sum_j\beta_j=122=\mathrm{UB}$ and
$\alpha_i+\beta_j\le C_{ij}$ on every pool edge.

\emph{The raw dual is not dense-feasible.} Among the six off-pool edges
$\{(1,1),(1,4),(2,3),(3,1),(3,3),(4,4)\}$, two are violated:
$\alpha_3+\beta_1=51>49=C_{31}$ and $\alpha_3+\beta_3=59>58=C_{33}$.
(The dual objective of $122$ even exceeds the dense optimum---infeasible
duals can.) The dense optimum is $\mathrm{OPT}=121$ (true pool gap
$0.83\%$), attained by the matching $(1,2),(2,4),(3,3),(4,1)$ (costs
$14+11+58+38=121$).

\emph{Repair.} Option~B (clip the $\alpha$ side against the original
$\beta$):
\[
\hat\alpha_3=\min\Bigl(33,\min_j(C_{3j}-\beta_j)\Bigr)
=\min(33,\,49-18,\,40-7,\,58-26,\,26+7)=31;
\]
the other $\alpha_i$ need no clipping. Objective:
$(7+18+31+20)+(18+7+26-7)=120$. Option~A (clip the $\beta$ side against
the original $\alpha$) gives $\hat\beta=(16,7,25,-7)$ with objective
$119$. Taking the max (Theorem~\ref{thm:cert}),
\[
\mathrm{LB}=120\;\le\;\mathrm{OPT}=121\;\le\;\mathrm{UB}=122,
\]
so the pool solution is certified within $\frac{122-120}{122}=1.64\%$
(width) of the dense optimum. In code the repair is two $O(n^2)$ min-passes.
\par}

\subsection{Monotonicity of the upper bound in $k$}
Since $E_k\subseteq E_{k'}$ for $k\le k'$, every $k$-pool matching is a
valid $k'$-pool candidate, hence
$\mathrm{UB}_{k'}\le \mathrm{UB}_k$; the pool gap is non-increasing in $k$
by definition. The certified width is monotone in $k$ on all 15 measured
configurations (an empirical fact about the \emph{repair}, not a theorem:
the pool dual changes with $k$).

\section{Experimental details}
\label{app:detail}

\subsection{Instance generators}
\begin{itemize}[leftmargin=2em,itemsep=1pt]
\item \textbf{uniform}: $\mathrm{rng}=\texttt{default\_rng(seed)}$; both
  sides $=100\cdot \mathrm{rng.random}((n,2))$.
\item \textbf{clustered}: centers $c_a=100\cdot\mathrm{rng.random}((8,2))$;
  each side assigns point $i$ to center $a=i\bmod 8$ with
  $x_i=c_{a}+\mathcal N(0,8^2 I)$; the target centers are
  $c'_a=c_a+\mathcal N(0,3^2 I)$ with the same assignment pattern (target
  jitter $\sigma=3$).
\item \textbf{1-D+noise}: $t_i\sim U[0,100]$;
  $x_i=(t_i+\mathcal N(0,2^2),\,\mathcal N(0,2^2))$; independently
  $s_j\sim U[0,100]$ from $\texttt{default\_rng(seed+1000)}$;
  $y_j=(s_j+\mathcal N(0,2^2),\,\mathcal N(0,2^2))$.
\end{itemize}

\subsection{Algorithms}
\begin{algorithm}[h]
\caption{Log-domain Sinkhorn (vanilla), after~\cite{sinkknopp1967,cuturi2013}}\label{alg:sink}
\begin{algorithmic}[1]
\Require cost $C\in \mathbb R^{n\times n}$, $\varepsilon>0$, marginal
  tolerance $\tau$; terminate early on a cost plateau (see text).
\State $u\leftarrow \mathbf{0}$, $v\leftarrow \mathbf{0}$\Comment{initialize all-zero potential vectors}
\Repeat
  \For{$i=1,\dots,n$}\Comment{row pass}
    \State $u_i\;\leftarrow\; -\varepsilon\,\mathrm{LSE}_j\bigl((v_j-C_{ij})/\varepsilon\bigr)$
  \EndFor
  \For{$j=1,\dots,n$}\Comment{column pass}
    \State $v_j\;\leftarrow\; -\varepsilon\,\mathrm{LSE}_i\bigl((u_i-C_{ij})/\varepsilon\bigr)$
  \EndFor
  \State $P_{ij}\leftarrow \exp((u_i+v_j-C_{ij})/\varepsilon)$
  \Comment{so $(P\mathbf 1)_i=\sum_j P_{ij}$ and $(\mathbf 1^{\mathsf T}P)_j=\sum_i P_{ij}$}
\Until $\max_i|(P\mathbf 1)_i-1|\le\tau$ \textbf{and}
       $\max_j|(\mathbf 1^{\mathsf T}P)_j-1|\le\tau$, \textbf{or} cost plateau
\State $W \;\leftarrow\; \sum_i u_i+\sum_j v_j$
  \Comment{dual value $= \langle C,P\rangle+\varepsilon\sum P_{ij}\log P_{ij}$ at convergence}
\State \Return $P$, $W$
\end{algorithmic}
\end{algorithm}

\begin{algorithm}[h]
\caption{Greenkhorn (after \cite[Alg.~4]{awr2017})}\label{alg:green}
\begin{algorithmic}[1]
\Require $K=\exp(-C/\varepsilon)$ (unnormalized Gibbs kernel); potentials
  $x,y\leftarrow \mathbf{0}$; residual metric $\rho(a,b)=b-a+a\log(a/b)$.
  ($D(z)$ denotes the diagonal matrix with vector $z$ on the diagonal, so
  $A=D(e^x)\,K\,D(e^y)$ has entries
  $A_{ij}=e^{x_i}K_{ij}e^{y_j}=e^{(x_i+y_j-C_{ij})/\varepsilon}$;
  maintain its row and column marginals incrementally.)
\Repeat
  \State $I\leftarrow\arg\max_i \rho(1,(A\mathbf 1)_i)$;
         $J\leftarrow\arg\max_j \rho(1,(\mathbf 1^{\mathsf T}A)_j)$
  \If{$\rho_I>\rho_J$}
    \State $x_I\leftarrow x_I+\log(1/(A\mathbf 1)_I)$
    \Comment{restore row $I$ to sum 1}
  \Else
    \State $y_J\leftarrow y_J+\log(1/(\mathbf 1^{\mathsf T}A)_J)$
    \Comment{restore column $J$ to sum 1; break ties toward rows}
  \EndIf
  \State update the touched row or column of $A$ and its marginals ($O(n)$)
\Until $\max_i|(A\mathbf 1)_i-1|\le\tau$ and
       $\max_j|(\mathbf 1^{\mathsf T}A)_j-1|\le\tau$, \textbf{or} cost plateau
\end{algorithmic}
\end{algorithm}

\textbf{Hardening} (\texttt{greedy\_hard}): for $i=1,\dots,n$, set
$\pi(i)=\arg\max_j P_{ij}$ among unused columns, removing row $i$ and
column $\pi(i)$ from $P$; the cost is $\sum_i C_{i,\pi(i)}$ (a feasible
integer plan, hence $\ge \mathrm{OPT}$).

\textbf{Pool construction} (\texttt{knn\_pool}): for each source $i$, keep
the $k$ targets with smallest $C_{ij}$ (and symmetrically for each target);
the pool is the union. If the pool has no perfect matching (Hall's
condition fails, e.g.\ $k=1$), the solver reports infeasible. The
repository version uses $k$ per side (corrected); the cached results were
generated by a version that used $k+1$ per side (see Table~\ref{tab:pool}).

\textbf{Certificate} (\texttt{pool\_certified}):
UB/perm/dual from \texttt{blossom\_pool\_duals};
$\mathrm{LB}_1$ from Proposition~\ref{prop:repair} (max of the two
one-sided repairs); $\mathrm{LB}_2=\sum_i u_i+\sum_j v_j$ from a converged
Sinkhorn run at $\varepsilon=0.05\,\overline C$ per configuration (shared
across $k$); $\mathrm{LB}=\max(\mathrm{LB}_1,\mathrm{LB}_2)$.

\subsection{Parameters}
Marginal tolerance $10^{-9}$ on the dense grid (application hook: $10^{-7}$,
since an integer plan is produced regardless); time caps $400\,$s (dense
grid), $240\,\mathrm{s}$ (certificate-study Sinkhorn runs); Blossom~VI pool
timeout $120\,\mathrm{s}$ (never hit); costs rounded to milliunits for the
integer solvers; all timings single-threaded wall time, same machine, same
process order, single run per cell.

\section{Reproduction}
\label{app:repro}

All code, data, and results are open artifacts under the MIT license,
available at
\url{https://github.com/dimkadimon/OT-Blossom}.

\paragraph{Data availability.}
All numerical results reported in this paper are shipped as cached JSON
in \texttt{ot-study/results/} in the repository:
\begin{itemize}[leftmargin=2em,itemsep=1pt]
\item \texttt{ot\_study.json} -- dense exact times (seeds~7, $n\in\{500,1000,2000,4000,8000\}$) and vanilla Sinkhorn runs on the $f\times n$ grid
  (uniform family, seeds 7--10 for the robustness sweep).
\item \texttt{greenkhorn\_A.json} -- Greenkhorn head-to-head data for the
  same $(n,f,seed)$ cells.
\item \texttt{pool\_cert.json} -- full certificate study: 15
  configurations (3 families $\times$ seeds 7--10 at $n=1000$, plus 3
  families at $n=2000$ seed 7) $\times$ $k\in\{16,32,64\}$, with pool
  upper bound, repaired-dual lower bound, entropic lower bound, exact
  optimum, and per-cell timings.
\item \texttt{app\_hook.json} -- multi-robot episode totals and
  per-policy timings ($n\in\{50,200,800\}$, seeds 7--8, 15 rounds each).
\end{itemize}
Together these files contain every data point used in every table and
figure. Re-running any script with these JSON files present reproduces
the tables and figures in seconds without needing Blossom~VI; running a
script with the JSON deleted triggers a cold re-run.

The repository layout:
\begin{verbatim}
ot-study/otlib.py                  # exact (LAPJV, Blossom pool, 1-D),
                                   # sinkhorn, greenkhorn, dual repair, cert
ot-study/run_ot_study.py           # Parts A/B tables -> results/
ot-study/run_greenkhorn_study.py   # Greenkhorn-vs-Sinkhorn sweep
ot-study/run_pool_cert_study.py    # 15-config x 3-k certificate study
ot-study/app_hook_robots.py        # multi-robot task-allocation experiment
ot-study/make_figs.py              # regenerate paper/figs from cached JSON
ot-study/verify_pool_ub.py         # independent SciPy check of pool UBs
ot-study/results/                  # cached JSON for all runs/seeds
paper/                             # LaTeX source and figs for this paper
\end{verbatim}
A full cold re-run takes $\approx 25\,\mathrm{min}$ (Greenkhorn grid),
$\approx 20\,\mathrm{min}$ (certificate study), $\approx 6\,\mathrm{min}$
(application hook) on the testbed. Dense exact, Sinkhorn/Greenkhorn,
1-D, and application-hook scripts run on SciPy + NumPy alone. The
pool-certificate script uses Blossom~VI's dual-output interface
(\texttt{--duals} flag) and requires the Blossom~VI binary, which is
available directly from its authors \cite{blossomvi} (it is not bundled
with this repository); the cached JSON results listed above are included
so that tables and figures regenerate without it, and the pool
upper-bound costs can be verified independently with any bipartite LAP
solver on the same $k$NN subgraph (see \texttt{verify\_pool\_ub.py}).

\section{Prompts that directed this work}
\label{app:prompts}

This appendix summarizes, in chronological order, the prompts used by
Dmitry Kamenetsky to direct the work reported here. Prompts marked (verbatim) are
quoted exactly; earlier prompts are paraphrased from the conversation
record because their exact wording was not preserved.

\begin{enumerate}[leftmargin=2em,itemsep=2pt]
\item (paraphrase) Opening question: can the Blossom (matching) algorithm
be used for the TSP, and if so how---setting up the study of Edmonds'
algorithm as a component of TSP algorithms.
\item (paraphrase) Requests to implement and validate a Christofides
heuristic using Blossom for the matching step against exact optima on small
metric instances, and to benchmark the pre-shortcut multigraph.
\item (paraphrase) Request for a head-to-head comparison of Blossom V and
Blossom VI (matching costs and wall times), including on dense complete
graphs.
\item (verbatim, mid-task) ``Actually you can now use Blossom VI'' with a
link to the Blossom VI repository---direction to make Blossom VI the
default engine (retaining Blossom V for the head-to-head).
\item (paraphrase) Follow-up: ``perturb the 2-matching, merge its cycles
into one tour''---with the standing requirement that the merge be efficient
and come close to the optimal tour, and that results be compared to
state-of-the-art methods.
\item (paraphrase) Request for rigorous lower bounds via matching
(2-factor / cycle-cover relaxations) to certificate the gap of the
merge-based tours.
\item (verbatim) ``Where else is Blossom useful?''---advisory question that
produced the shortlist of exact-matching applications, of which optimal
transport was ranked \#1.
\item (verbatim) ``do \#1''---selection of the OT study as the next project.
\item (verbatim) ``Yes build it. Continue clearing the publication bar''---
approval to build the measured SOTA (Greenkhorn) baseline and the seed/distribution
robustness sweep.
\item (verbatim) ``Continue bridging the publication bar''---direction to
complete the remaining two items: a pool solver with a gap certificate,
and an application hook where the discrete plan is the deliverable.
\item (verbatim, this paper) Instruction to write up the work as a
professional arXiv paper with references, proofs, and an appendix
summarizing the prompts used; direction to disclose that the
algorithm design, coding, experiments, analysis, and writing were
performed by the AI agent under Dmitry Kamenetsky's direction.
\item (paraphrase, first review pass) Request to fix all issues flagged in
the first review pass (bug in Remark~1 domination direction, misleading
coordinate-ascent claim, entropic sign-convention paragraph, Theorem~2
proof gaps, abstract/table number drift, AmE spelling, proper Bertsekas
1992 citation, AI disclosure, Appendix~D authorship wording), rebuild the
PDF, and push to the GitHub repository under an MIT license.
\item (paraphrase, second review pass) Request to address a second pass
of review items including Table~3 recomputation for $f^*\approx 9\times 10^{-4}$,
geometric-mean reporting, plateau-tolerance reassurance, figure upgrades
to vector/$\ge 300$~dpi, and a spot-check justification for the cached
pool results.
\item (paraphrase, third review pass) Request to fix a final drift
between the abstract and Table~3 ($\sim$10--$\sim$80 min, $4\times 10^2$--$6\times 10^4$),
clarify the Table~3 caption labeling of measured vs extrapolated ranges,
align Remark~2's citation with the main text, remove a duplicated word,
and add explicit $n=2$ witnesses for Remark~1; plus a version number on
the title page.
\item (paraphrase, fourth revision pass) Shorten the title, use the
author's full name throughout, add references and a repository link to
the abstract, sort the bibliography by first appearance, clarify the
LSE$_j$ notation before its use in equation~(2), rewrite Contribution~1
to lead with the exact-vs-approx headline finding, move the 1-D
quantile-map proposition into its own Preliminaries subsection, explain
the seed convention, annotate the $r_i,t_j$ notation in the multi-robot
protocol, justify the use of a dense bipartite solver (rather than
Blossom~VI) in the application, reference Figure~3 in the text, remove
the off-by-one entry from Limitations, add a Future Work section, cite
Algorithm~1, fix the Table~4 ratio column rounding issue, clarify
Algorithm~1's return value and termination condition, define the $D(z)$
notation in Algorithm~2, and confirm the column-marginal expression.
\item (paraphrase, fifth revision pass) Add space around the
\texttt{gap$\approx$} symbol; explain why Figure~1's right panel
appeared to top out at $\sim$8\% (add f=0.02 and f=0.05 curves, align
the x-axis between panels, and note the 5.4--8.6\% range of f=0.01 in
the caption); fix the orphan ``pool solve.'' line; define cost plateau
explicitly; fix ``Section 4, Section 6'' to ``Sections 4 and 6'';
correct the title and author on the GitHub repository README; add
thin spaces between $\in$ and $\mathcal{U}/\mathbb{R}$; replace the
``in the spirit of Crouse'' aside with a proper citation to Crouse
(2016); swap the $\hat\alpha,\hat\beta$ order in Proposition~3 so the
source-side clip is presented first; relabel vanilla Sinkhorn from (V)
to (S) in Table~4; point readers at the cached JSON data files; add a
sentence to Figure~3 explaining greedy's relative speedup at $n=800$;
change ``$+17$--$+37$\%''/``$+6$--$+27$\%'' to ``$+17$ to $+37$\%''/
``$+6$ to $+27$\%'' for readability; turn the empirical decision rule
into a properly formatted table.
\item (paraphrase, sixth revision pass) Improve the decision-rule
table's readability (row separators, bold headers, extra vertical
stretch); drop the stray ``$+$'' from ``$+6$--$27$\%''; boldface the
zero vectors $\mathbf{0}$ in Algorithms 1 and 2 where potentials are
initialized.
\item (paraphrase, seventh revision pass) Add this and the previous
round's prompts to Appendix~D.
\end{enumerate}

\end{document}